\documentclass[11pt]{article}

\usepackage[a4paper,margin=1in]{geometry}
\usepackage[T1]{fontenc}
\usepackage[utf8]{inputenc}
\usepackage{microtype}
\usepackage{mathtools,amssymb,amsthm}
\allowdisplaybreaks[2]
\usepackage{enumitem}
\usepackage[numbers,square]{natbib}
\usepackage{xurl}
\usepackage[hidelinks]{hyperref}
\usepackage[nameinlink,capitalise,noabbrev]{cleveref}

\hypersetup{
  pdftitle={Randomizing the Number of Centers in k-means++},
  pdfauthor={Vaclav Rozhon},
  pdfsubject={Budget-smoothed approximation guarantees for k-means++ seeding},
  pdfkeywords={k-means++, cost-proportional sampling, D-squared sampling, budget-smoothed analysis, random cardinality}
}

\newcommand{\X}{X}
\newcommand{\R}{\mathbb{R}}
\newcommand{\E}{\mathbb{E}}
\newcommand{\Prob}{\mathbb{P}}
\newcommand{\OPT}{\operatorname{OPT}}
\newcommand{\ALG}{\operatorname{ALG}}
\newcommand{\one}{\mathbf{1}}
\newcommand{\cU}{\mathcal{U}}
\newcommand{\cW}{\mathcal{W}}

\renewcommand{\phi}{\varphi}

\newtheorem{theorem}{Theorem}
\newtheorem{lemma}[theorem]{Lemma}

\theoremstyle{definition}

\title{Randomizing the Number of Centers in $k$-means++}
\author{V\'aclav Rozho\v{n}}
\date{}

\begin{document}
\maketitle

\begin{abstract}
The $k$-means++ algorithm is a standard and widely used seeding method for
$k$-means clustering, but for a fixed number $k$ of centers its worst-case
expected approximation ratio is $\Theta(\log k)$. We consider the same
algorithm when an adversary first fixes the dataset and some $K$; the number
of centers $k$ is then chosen uniformly from $\{K,\ldots,2K-1\}$. We prove
that $k$-means++ is an
$O(1)$-approximation with constant probability in this
\emph{budget-smoothed} setup.
\end{abstract}

\section{Introduction}\label{sec:intro}

The $k$-means problem asks us to choose a set $C=\{c_1,\dots,c_k\}$ of $k$
centers for a dataset $\X$ of $n$ points in $\R^d$ so as to minimize
\[
  \phi(\X,C)
  =\sum_{x\in\X}\min_{c\in C}\|x-c\|_2^2.
\]
A classical initialization for
this problem is the ordinary $k$-means++ algorithm of
\cite{arthur2007kmeanspp}; a greedy variant of this initializer is used,
for example, in the scikit-learn library
\citep{pedregosa2011scikit}. The ordinary algorithm chooses the first center
uniformly from the input points and, given a set of already sampled
centers $C_i$, every subsequent center $c_{i+1}$ is chosen with
probability proportional to its current contribution to the cost. That is,
each point $x\in\X$ is chosen with probability
$\phi(x,C_i)/\phi(\X,C_i)$. We call such a draw a
\emph{cost-proportional sample}; it is often called a $D^2$-sample in the
literature.

If the algorithm chooses exactly $k$ centers, its expected cost is
$O(\log k)$ times the optimum, and this is tight
\citep{arthur2007kmeanspp}.
The known lower bounds, however, choose an instance for a prescribed value
of $k$. They do not show that one fixed instance is simultaneously hard
for many nearby values of $k$.

In this paper we consider the following experiment. An adversary first fixes $\X$ and $K$. We then choose
\begin{equation}\label{eq:experiment}
  k\sim\operatorname{Unif}(\{K,\ldots,2K-1\}),
\end{equation}
run the ordinary $k$-means++ algorithm for $k$ steps, and compare its cost with the optimum using the same number $k$ of centers. 

We show that for this smoothed experiment, the $k$-means++ algorithm in fact achieves $O(1)$ competitive ratio with constant probability. In particular, there is always a good set $G \subseteq \{K, ..., 2K-1\}$ on which the algorithm achieves a constant approximation ratio. 

\begin{theorem}\label{thm:main}
There is a universal constant $C$ such that, for every finite nonempty
dataset $\X$ and every $K\ge1$, there is a set
\[
  G\subseteq\{K,\ldots,2K-1\},
  \qquad |G|> K/2,
\]
such that every $t\in G$ satisfies
\[
  \Prob \left[
    \frac{\ALG^t(\X)}{\OPT^t(\X)}\le C
  \right]\ge\frac12.
\]
\end{theorem}

Here $\ALG^t(\X)$ is the cost of the first $t$ sampled centers and
$\OPT^t(\X)=\min_{|C|=t}\phi(\X,C)$; the probability is over the randomness in the algorithm.
Averaging over $t$ immediately gives a $C$-approximation with joint probability greater than $1/4$ in \eqref{eq:experiment}.

The parameters in the theorem are somewhat arbitrary -- one can replace the range $\{K, ..., 2K-1\}$ with any range of size $\Omega(K)$ and the proportion $|G|/K$ and success probability can be made arbitrarily close to $1$ at the expense of the constant $C$ that we do not try to optimize.

We believe our experiment matches the practical usage of $k$-means -- for
nontrivial-sized $k$ where the $O(1)$ vs $O(\log k)$ distinction matters,
we typically do not have a precise value of $k$ in mind, we simply aim to
discretize our original dataset $\X$ into a much smaller $k$-sized one.
\Cref{thm:main} then says that for most $k$'s in our intended range, the
algorithm is $O(1)$-approximate.

\subsection{Related work}

\paragraph{Fixed-budget guarantees.}
\cite{arthur2007kmeanspp} proved the original
$8(\ln k+2)$ expected approximation bound for $k$-means++ and a matching
$\Omega(\log k)$ lower bound. The upper bound was sharpened to
$5(\ln k+2)$ by \cite{makarychev2020improved}. A logarithmic loss can
occur with very high probability \citep{brunsch2013bad}, and already in two
dimensions \citep{bhattacharya2016tight}. These results fix the target
budget before constructing the instance.
\cite{ostrovsky2006effectiveness} obtained stronger guarantees under a
separation condition expressed by a sufficiently large drop between
successive optimum values.

\paragraph{Oversampling and variants of $k$-means++ seeding.}
If the algorithm may output more centers than the comparator, additional
centers sampled by the $k$-means++ rule improve the guarantee
\citep{aggarwal2009adaptive,wei2016constant,makarychev2020improved}.
In particular, $k+\Delta$ samples give an
$O(\log(2k/\Delta))$ approximation to $\OPT^k(\X)$ in the small-oversampling
regime \citep{makarychev2020improved}. This is the intuition behind our
proof, but we cannot apply the result directly: we always
compare $\ALG^t(\X)$ with $\OPT^t(\X)$ in our result.
Related analyses concern parallel
\citep{bahmani2012scalable,rozhon2020simple}, local-search
\citep{lattanzi2019better,choo2020few}, greedy and noisy
\citep{bhattacharya2020noisy,grunau2023greedy,grunau2023noisy}, and
outlier-robust \citep{grunau2022outliers} variants of $k$-means++.

\paragraph{Smoothed and incremental models.}
Classical smoothed analysis perturbs an adversarial numerical instance
\citep{spielman2004smoothed}; for clustering,
\cite{arthur2011smoothed} used this framework to analyze the running time
of Lloyd's method. Our geometry is not perturbed. A closely related model
is the budget-smoothed analysis of submodular maximization
\citep{rubinstein2022budget}; the same viewpoint has also been used in
budget-feasible mechanism design \citep{rubinstein2023beyond}. A different
line of work constructs, for metric $k$-median and related objectives, one
nested sequence of centers that is competitive at every cardinality
\citep{mettu2003online,chrobak2008incremental}. Applying \cref{thm:main} for all $K = 2^i$, we show that the nested sequence of the ordinary randomized $k$-means++ works on all but a constant fraction of indices.  

\subsection{Our method, in a nutshell}

The main idea behind our proof is the following. Consider two special cases. First, consider the case where the optimal cost for $K$ clusters is comparable for the optimal cost for $2K$ clusters. In that case, we can use the oversampling literature that proves that $k$-means++ with $2K$ centers is $O(1)$ approximation of optimum with $K$ centers to conclude that for most values of $k \in \{K, ..., 2K-1\}$, $k$-means++ cost is comparable to optimal cost on the same number of clusters. 

In the opposite case, optimum solution drops by, say, the same constant factor whenever we increase the number of centers from $k$ to $k+1$. In that scenario however, we can analyze the $k$-means++ algorithm more diligently and observe that throughout the algorithm, it behaves in a very pleasant way -- each new sampled center with high probability hits a new optimal cluster. This leads to the algorithm finishing with a good approximation with large constant probability. 

Our analysis merges the two cases by splitting the range $\{K, ..., 2K-1\}$ into blocks, inside each of which the cost of the optimum solution remains the same. Each block is analyzed by oversampling argument, and in between blocks, the algorithm can be controlled akin our second special case. 

\section{Preliminaries}\label{sec:preliminaries}

We regard $\X$ as a multiset of $n$ points of $\R^d$, so that repeated
locations remain distinct elements; sampling a point of $\X$ means sampling
one of its elements. For $P\subseteq\X$ and a finite center set $C\subseteq\R^d$,
write
$\phi(P,C)=\sum_{x\in P}\min_{c\in C}\|x-c\|_2^2$; we write $\phi(x,C)$ for $\phi(\{x\},C)$.

The $k$-means++ algorithm first samples $c_1$ uniformly from $\X$. Next, each $c_i, i \ge 2$ is sampled from $\X$ \emph{proportionally to its cost}, i.e., $x \in \X$ is picked with probability $\phi(x, C_{i-1}) / \phi(\X, C_{i-1})$, where we use $C_i = \{c_1, ..., c_i\}$.

We use $\OPT^t(P) = \min_{c_1^\star, ..., c_t^\star \in \R^d} \phi(P, \{c_1^\star, ..., c_t^\star\})$.
It is a standard fact that for $\OPT^1(P)$, the minimizer $c^\star$ is the mean of the pointset $\frac{1}{|P|} \sum_{x\in P} x$. For each $t$, we fix one optimal clustering of $\X$ and denote it as $\mathcal{P}^t = (P_1^t, ..., P_t^t)$; in particular we have $\OPT^t(\X) = \OPT^1(P_1^t) + ... + \OPT^1(P_t^t)$.

We use $\ALG^t(P)$ to denote the cost $\phi(P, C_t)$ of the first $t$ centers $c_1, ..., c_t$ generated by $k$-means++; this cost is a random variable. Since $C_t$ is a feasible solution with $t$ centers,
\begin{equation}\label{eq:prefix-lower}
  \ALG^t(P)\ge\OPT^t(P).
\end{equation}

In the analysis below, we assume $\OPT^{2K-1}(\X)>0$, which ensures that
every required sampling distribution is well-defined. The general case is
handled by truncating the optimum curve at its last positive value and
treating all subsequent budgets as optimal; we omit these routine details.

We use the following standard
one-cluster estimate of \cite{arthur2007kmeanspp} with tight constants from  \cite{makarychev2020improved}.

\begin{lemma}[One-cluster estimates]\label{lem:one-cluster}
Let $P\subseteq\X$ be nonempty.
\begin{enumerate}[label=(\roman*)]
  \item If $c$ is uniform in $P$, then
  \[
    \E \left[ \phi(P, \{c\}) \right] =2\OPT^1(P).
  \]
  \item Fix arbitrary current centers $C_i$, and pick $c_{i+1} \in P$ proportionally to its cost, i.e., $x \in P$ is picked with probability $\phi(x, C_i) / \phi(P, C_i)$. Then,
  \[
    \E \left[ \phi(P, C_i \cup \{c_{i+1}\}) \right] \le 5\OPT^1(P).
  \]
\end{enumerate}
\end{lemma}

\section{Analysis}\label{sec:analysis}

\subsection{Wasted centers in a prefix}\label{sec:wasted}

Recall that $c_1, c_2, \dots$ are the centers sampled by $k$-means++. Fix a
reference clustering $\mathcal P^t$. For a step $i$, we call a cluster $P_j^t$
\emph{covered} if $C_i\cap P_j^t\neq\emptyset$ and \emph{uncovered} otherwise;
this splits the clustering into
\[
  \mathcal P^t=\cU_i^t\sqcup\cW_i^t,
\]
where $\cU_i^t$ collects the uncovered clusters and $\cW_i^t$ the covered
ones. Throughout, we abuse the notation and write $\phi(\cW_i^t, C)$ for
$\phi\big(\bigcup_{P\in\cW_i^t}P,\;C\big)$, i.e., we identify a collection of
clusters with the set of points it contains. In particular,
\[
  \phi(X, C_i) =\phi(\cU_i^t,C_i)+\phi(\cW_i^t,C_i).
\]

We observe that the expected cost of the covered clusters is small, as in the
standard $k$-means++ analysis \citep{arthur2007kmeanspp}.

\begin{lemma}\label{lem:first-hit}
For every $i$ and $t$ we have $\E[\phi(\cW_i^t, C_i)]\le5\OPT^t(\X)$.
\end{lemma}

\begin{proof}
Fix a covered cluster $P_j^t$ and consider the step $i'\le i$ when we sampled
$c_{i'}\in P_j^t$. Applying \cref{lem:one-cluster}, we conclude that
$\E[\phi(P_j^t, C_{i'})]\le5\OPT^1(P_j^t)$. Since
$\phi(P_j^t, C_i)\le\phi(P_j^t, C_{i'})$, summing over the covered clusters
gives $\E[\phi(\cW_i^t, C_i)]\le5\sum_j\OPT^1(P_j^t)=5\OPT^t(\X)$.
\end{proof}

We call a center $c_{i+1}$ \emph{wasted} relative to $\mathcal P^t$ if it lands
in an already covered cluster, i.e., if $c_{i+1}\in\cW_i^t$. Our goal is to
upper bound the number of wasted centers; let $R^t$ be their number among
$c_2,\ldots,c_t$. To this end, we define
\begin{equation}\label{eq:zq}
  z^t=\sum_{i=1}^{t-1}\frac{\OPT^t(\X)}{\OPT^i(\X)}.
\end{equation}

\begin{lemma}\label{lem:wasted}
For every $t$, we have
\[
  \E\left[R^t\right]\le5z^t.
\]
\end{lemma}

\begin{proof}
Given the first $i$ samples, the center $c_{i+1}$ is sampled proportionally to
its cost, so it is wasted with probability $\phi(\cW_i^t,C_i)/\phi(\X,C_i)$.
By \eqref{eq:prefix-lower}, $\phi(\X,C_i)=\ALG^i(\X)\ge\OPT^i(\X)$, and
therefore, by \cref{lem:first-hit},
\[
  \Prob(c_{i+1}\text{ is wasted})
  \le\frac{\E[\phi(\cW_i^t, C_i)]}{\OPT^i(\X)}
  \le5\frac{\OPT^t(\X)}{\OPT^i(\X)}.
\]
Summing over $i=1,\ldots,t-1$ proves the claim.
\end{proof}

\subsection{Using the additional samples}\label{sec:additional}

The next lemma uses $z^q$ to bound the probability that the cost of
$k$-means++ is substantially larger than the optimum. Allowing $e$
additional samples improves this bound.

\begin{lemma}\label{lem:repair}
For every $q\ge1$, $e\ge0$, and $\Lambda>0$,
\begin{equation}\label{eq:repair-tail}
  \Prob[\ALG^{q+e}(\X)>\Lambda\OPT^q(\X)]
  =O\left(\frac{z^q}{e+1}+\frac1\Lambda\right).
\end{equation}
\end{lemma}

\begin{proof}
Throughout the proof we use Lemma~\ref{lem:wasted} and the notation of
\cref{sec:wasted} with the reference clustering $\mathcal P^q$.
After the first $q$ samples, every non-wasted sample has covered a new
cluster, so
\[
  |\cU_q^q|=R^q.
\]

Let $H$ be the event that at some point during the next $e$ samples the
uncovered clusters cost no more than the covered ones. Until this happens,
the next sample covers a new cluster with probability greater than $1/2$.
We can therefore compare the $e$ samples with $e$ independent fair coin
tosses, coupled so that every head before the stopping condition defining
$H$ is met means that a new cluster is covered.

Let $B\sim\operatorname{Bin}(e,1/2)$ be the number of heads. If $H$ does
not happen, then fewer than the $R^q$ initially uncovered clusters have
been covered, and hence $B<R^q$. Conditional on the first $q$ samples,
$R^q$ is fixed and $B$ is independent of this prefix. Since
$\one_{\{B<R^q\}}\le R^q/(B+1)$, we get
\[
  \Prob(H^c\mid c_1,\ldots,c_q)
  \le R^q\,\E\frac1{B+1}.
\]
Using $1/(m+1)=\int_0^1u^m\,du$ and the binomial generating function,
\[
\begin{aligned}
  \E\frac1{B+1}
\ &=\int_0^1\E[u^B]\,du\\
  &=\int_0^1\left(\frac{1+u}{2}\right)^e du\\
  &=\frac{2(1-2^{-(e+1)})}{e+1}
  \le\frac2{e+1}.
\end{aligned}
\]
Averaging over the prefix $c_1,\ldots,c_q$ and applying
Lemma~\ref{lem:wasted} gives
\begin{equation}\label{eq:repair-fail}
  \Prob(H^c)
  \le\frac{2\E R^q}{e+1}
  \le\frac{10z^q}{e+1}
  =O\left(\frac{z^q}{e+1}\right).
\end{equation}

On $H$, consider the first step at which the uncovered cost is at most the
covered cost. At that step the total cost is at most twice the covered
cost, and adding more centers can only decrease it. Charge each covered
cluster its cost just after it was hit for the first time. By
\cref{lem:one-cluster}, the expected sum of these charges is at most
$5\OPT^q(\X)$, while the current cost of a covered cluster is at most its
charge. It follows that
\[
  \E\!\left[
    \ALG^{q+e}(\X)\one_H
  \right]
  =O(\OPT^q(\X)).
\]
Markov's inequality now yields
\[
\begin{aligned}
  \Prob[\ALG^{q+e}(\X)>\Lambda\OPT^q(\X)]
  &\le\Prob(H^c)+\Prob\bigl(
    H,\ \ALG^{q+e}(\X)>\Lambda\OPT^q(\X)
  \bigr)\\
  &=O\left(\frac{z^q}{e+1}+\frac1\Lambda\right).
\end{aligned}
\]
This proves \eqref{eq:repair-tail}.
\end{proof}

\subsection{Choosing the reference prefix}\label{sec:shifted}

We partition the indices $\{1,\ldots,2K-1\}$ into
blocks on which the optimum costs are within a constant factor. For a
target budget $t$, we use the first index $q$ of its block as the reference
prefix. The $t-q$ additional samples allow \cref{lem:repair} to compare
$\ALG^t$ with $\OPT^q$, while the block construction compares $\OPT^q$
with $\OPT^t$. The following deterministic lemma gives such a partition
for any decreasing sequence.

\begin{lemma}\label{lem:shifted}
Let
\[
  a_1\ge a_2\ge\cdots\ge a_N>0,
  \qquad
  z_q=\sum_{i<q}\frac{a_q}{a_i}.
\]
For every $L>0$, there is a partition of
$\{1,\ldots,N\}$ into consecutive blocks with the following properties.
If $p(t)$ is the first index in the block containing $t$, then
\begin{equation}\label{eq:within-block}
  \frac{a_{p(t)}}{a_t}<e^L
\end{equation}
for every $t$, and at most
\begin{equation}\label{eq:number-bad}
  \frac{N-1}{4}
\end{equation}
indices satisfy
\begin{equation}\label{eq:bad}
  t-p(t)+1<\frac{L}{4}z_{p(t)}.
\end{equation}
\end{lemma}

\begin{proof}
Put
\[
  v_i=\log\frac{a_1}{a_i},
\]
so $0=v_1\le\cdots\le v_N$. Choose $\theta$ uniformly from $[0,L)$ and
partition the indices into maximal consecutive blocks on which
\[
  b_\theta(i)=
  \left\lfloor\frac{v_i+\theta}{L}\right\rfloor
\]
is the same value. If $p$ and $t$ lie in one block, then
$0\le v_t-v_p<L$ which implies the required property \eqref{eq:within-block}.

For $p\ge2$, write $\Delta_p=v_p-v_{p-1}$. Note that we have
\begin{equation}\label{eq:start-probability}
  \Prob_\theta(p\text{ begins a block})
  =\min\left\{1,\frac{\Delta_p}{L}\right\}.
\end{equation}
Index $1$ always begins a block, but $z_1=0$. Expanding $z_p$ and
exchanging the sums gives
\[
  \E_\theta\sum_{p\text{ begins a block}}z_p
  =\sum_{i=1}^{N-1}\sum_{p=i+1}^{N}
   e^{-(v_p-v_i)}
   \min\left\{1,\frac{\Delta_p}{L}\right\}.
\]
For fixed $i$ and $p>i$,
\[
  e^{-(v_p-v_i)}
  \min\left\{1,\frac{\Delta_p}{L}\right\}
  \le
  \frac{e^{-(v_{p-1}-v_i)}-e^{-(v_p-v_i)}}{L}.
\]
This holds because after multiplying by $Le^{v_p-v_i}$, we get
$L\min\{1,\Delta_p/L\}\le e^{\Delta_p}-1$. Summing over
$p=i+1,\ldots,N$ telescopes to at most $1/L$. It follows that
\begin{equation}\label{eq:start-weight}
  \E_\theta\sum_{p\text{ begins a block}}z_p
  \le\frac{N-1}{L}.
\end{equation}

In a block beginning at $p$, fewer than $(L/4)z_p$ positive integers
$s=t-p+1$ satisfy $s<(L/4)z_p$. Hence the number $B_\theta$ of indices
satisfying \eqref{eq:bad} satisfies
\[
  B_\theta
  \le\frac{L}{4}\sum_{p\text{ begins a block}}z_p.
\]
By \eqref{eq:start-weight}, some shift has
$B_\theta\le(N-1)/4$.
\end{proof}

\subsection{Proof of the main theorem}\label{sec:main-proof}

\begin{proof}[Proof of \cref{thm:main}]
Recall from \cref{sec:preliminaries} that we may assume
$\OPT^{2K-1}(\X)>0$. Apply \cref{lem:shifted} with
\[
  N=2K-1,\qquad
  a_i=\OPT^i(\X)\quad(1\le i\le N).
\]
Then the quantity $z_q$ of the lemma is our $z^q$ from \eqref{eq:zq}.
Choose a sufficiently large universal constant $L$ and set $\Lambda=L$.
The number of bad indices is at most
\[
  \frac{N-1}{4}
  =\frac{K-1}{2},
\]
which is smaller than $K/2$. Hence more than $K/2$ values
$t\in\{K,\ldots,2K-1\}$ are good.

Fix such a $t$, let $q=p(t)$ be the first index of its block, and put
$e=t-q$. Goodness and \eqref{eq:within-block} give
\[
  e+1\ge\frac{L}{4}z^q,
  \qquad
  \frac{\OPT^q(\X)}{\OPT^t(\X)}<e^{L}.
\]
By \cref{lem:repair}, the probability that $\ALG^t(\X)>\Lambda\OPT^q(\X)$ is
at most
\[
  O\left(\frac{z^q}{e+1}+\frac1\Lambda\right)
  =O\left(\frac4L+\frac1\Lambda\right)
  =O\left(\frac1L\right)
  <\frac12.
\]
Thus every good $t$ satisfies
\[
  \Prob\!\left[
    \frac{\ALG^t(\X)}{\OPT^t(\X)}
    \le\Lambda e^{L}
  \right]>\frac12,
\]
which proves \cref{thm:main} with $C=\Lambda e^{L}$. \end{proof}

\section{Concluding remarks}\label{sec:conclusion}

\paragraph{Other budget distributions.}
The proof never uses uniformity of the budget beyond an upper bound on the
probability of a single value. For example, if $k$ is drawn from a
geometric distribution, the same argument gives a constant approximation
ratio with constant probability.

\paragraph{Guarantee in expectation}
We do not know whether the analysis also holds in expectation. That is, whether \[
  \E\left[\frac{\ALG^k(\X)}{\OPT^k(\X)}\right]
  =O(1)
\]
where the expectation is both over the choice of $k$ and the randomness of the algorithm. 

\paragraph{Acknowledgments.}
The author developed the initial $O(\log\log k)$ argument. GPT-5.6
completed the remaining steps leading to the $O(1)$-approximation result.
Claude Opus 5 assisted with the exposition. The author takes responsibility
for the contents of the paper.

A Lean formalization of \cref{thm:main} is available in a public GitHub
repository \citep{rozhon2026lean}.

\begingroup
\small
\setlength{\bibsep}{2pt plus .5pt}
\bibliographystyle{alpha}
\bibliography{random_k_kmeanspp}

\newcommand{\etalchar}[1]{$^{#1}$}
\begin{thebibliography}{BMV{\etalchar{+}}12}

\bibitem[ADK09]{aggarwal2009adaptive}
Ankit Aggarwal, Amit Deshpande, and Ravi Kannan.
\newblock Adaptive sampling for $k$-means clustering.
\newblock In {\em Approximation, Randomization, and Combinatorial Optimization:
  Algorithms and Techniques}, volume 5687 of {\em Lecture Notes in Computer
  Science}, pages 15--28. Springer, 2009.

\bibitem[AMR11]{arthur2011smoothed}
David Arthur, Bodo Manthey, and Heiko R{\"o}glin.
\newblock Smoothed analysis of the {$k$}-means method.
\newblock {\em Journal of the ACM}, 58(5):19:1--19:31, 2011.

\bibitem[AV07]{arthur2007kmeanspp}
David Arthur and Sergei Vassilvitskii.
\newblock {$k$-means++}: The advantages of careful seeding.
\newblock In {\em Proceedings of the 18th Annual ACM-SIAM Symposium on Discrete
  Algorithms}, pages 1027--1035. SIAM, 2007.

\bibitem[BERS20]{bhattacharya2020noisy}
Anup Bhattacharya, Jan Eube, Heiko R\"oglin, and Melanie Schmidt.
\newblock Noisy, greedy and not so greedy $k$-means++.
\newblock In {\em 28th Annual European Symposium on Algorithms}, volume 173 of
  {\em LIPIcs}, pages 18:1--18:21, 2020.

\bibitem[BJA16]{bhattacharya2016tight}
Anup Bhattacharya, Ragesh Jaiswal, and Nir Ailon.
\newblock Tight lower bound instances for $k$-means++ in two dimensions.
\newblock {\em Theoretical Computer Science}, 634:55--66, 2016.

\bibitem[BMV{\etalchar{+}}12]{bahmani2012scalable}
Bahman Bahmani, Benjamin Moseley, Andrea Vattani, Ravi Kumar, and Sergei
  Vassilvitskii.
\newblock Scalable $k$-means++.
\newblock {\em Proceedings of the VLDB Endowment}, 5(7):622--633, 2012.

\bibitem[BR13]{brunsch2013bad}
Tobias Brunsch and Heiko R\"oglin.
\newblock A bad instance for $k$-means++.
\newblock {\em Theoretical Computer Science}, 505:19--26, 2013.

\bibitem[CGPR20]{choo2020few}
Davin Choo, Christoph Grunau, Julian Portmann, and V{\'a}clav Rozho{\v{n}}.
\newblock {$k$-means++}: Few more steps yield constant approximation.
\newblock In {\em Proceedings of the 37th International Conference on Machine
  Learning}, volume 119 of {\em Proceedings of Machine Learning Research},
  pages 1909--1917. PMLR, 2020.

\bibitem[CKNY08]{chrobak2008incremental}
Marek Chrobak, Claire Kenyon, John Noga, and Neal~E. Young.
\newblock Incremental medians via online bidding.
\newblock {\em Algorithmica}, 50(4):455--478, 2008.

\bibitem[G{\"O}R23]{grunau2023noisy}
Christoph Grunau, Ahmet~Alper {\"O}z{\"u}do{\u{g}}ru, and V{\'a}clav
  Rozho{\v{n}}.
\newblock Noisy {$k$-Means++} revisited.
\newblock In {\em 31st Annual European Symposium on Algorithms (ESA 2023)},
  volume 274 of {\em Leibniz International Proceedings in Informatics
  (LIPIcs)}, pages 55:1--55:7. Schloss Dagstuhl--Leibniz-Zentrum f{\"u}r
  Informatik, 2023.

\bibitem[G{\"O}RT23]{grunau2023greedy}
Christoph Grunau, Ahmet~Alper {\"O}z\"udo\u{g}ru, V\'aclav Rozho\v{n}, and
  Jakub T\v{e}tek.
\newblock A nearly tight analysis of greedy $k$-means++.
\newblock In {\em Proceedings of the 2023 Annual ACM-SIAM Symposium on Discrete
  Algorithms}, pages 1012--1070. SIAM, 2023.

\bibitem[GR22]{grunau2022outliers}
Christoph Grunau and V{\'a}clav Rozho{\v{n}}.
\newblock Adapting {$k$-means} algorithms for outliers.
\newblock In {\em Proceedings of the 39th International Conference on Machine
  Learning}, volume 162 of {\em Proceedings of Machine Learning Research},
  pages 7845--7886. PMLR, 2022.

\bibitem[LS19]{lattanzi2019better}
Silvio Lattanzi and Christian Sohler.
\newblock A better {$k$-means++} algorithm via local search.
\newblock In {\em Proceedings of the 36th International Conference on Machine
  Learning}, volume~97 of {\em Proceedings of Machine Learning Research}, pages
  3662--3671. PMLR, 2019.

\bibitem[MP03]{mettu2003online}
Ramgopal~R. Mettu and C.~Greg Plaxton.
\newblock The online median problem.
\newblock {\em SIAM Journal on Computing}, 32(3):816--832, 2003.

\bibitem[MRS20]{makarychev2020improved}
Konstantin Makarychev, Aravind Reddy, and Liren Shan.
\newblock Improved guarantees for $k$-means++ and $k$-means++ parallel.
\newblock In {\em Advances in Neural Information Processing Systems 33}, 2020.

\bibitem[ORSS06]{ostrovsky2006effectiveness}
Rafail Ostrovsky, Yuval Rabani, Leonard~J. Schulman, and Chaitanya Swamy.
\newblock The effectiveness of {Lloyd}-type methods for the $k$-means problem.
\newblock In {\em Proceedings of the 47th Annual IEEE Symposium on Foundations
  of Computer Science}, pages 165--176, 2006.

\bibitem[PVG{\etalchar{+}}11]{pedregosa2011scikit}
Fabian Pedregosa, Ga{\"e}l Varoquaux, Alexandre Gramfort, Vincent Michel,
  Bertrand Thirion, Olivier Grisel, Mathieu Blondel, Peter Prettenhofer, Ron
  Weiss, Vincent Dubourg, Jake Vanderplas, Alexandre Passos, David Cournapeau,
  Matthieu Brucher, Matthieu Perrot, and {\'E}douard Duchesnay.
\newblock Scikit-learn: Machine learning in {Python}.
\newblock {\em Journal of Machine Learning Research}, 12(85):2825--2830, 2011.

\bibitem[Roz20]{rozhon2020simple}
V{\'a}clav Rozho{\v{n}}.
\newblock Simple and sharp analysis of {$k$-means{\textbar}{\textbar}}.
\newblock In {\em Proceedings of the 37th International Conference on Machine
  Learning}, volume 119 of {\em Proceedings of Machine Learning Research},
  pages 8266--8275. PMLR, 2020.

\bibitem[Roz26]{rozhon2026lean}
V{\'a}clav Rozho{\v{n}}.
\newblock Randomizing the number of centers in {$k$}-means++: {Lean}
  formalization, 2026.
\newblock GitHub repository.

\bibitem[RZ22]{rubinstein2022budget}
Aviad Rubinstein and Junyao Zhao.
\newblock Budget-smoothed analysis for submodular maximization.
\newblock In Mark Braverman, editor, {\em 13th Innovations in Theoretical
  Computer Science Conference (ITCS 2022)}, volume 215 of {\em Leibniz
  International Proceedings in Informatics (LIPIcs)}, pages 113:1--113:23.
  Schloss Dagstuhl--Leibniz-Zentrum fuer Informatik, 2022.

\bibitem[RZ23]{rubinstein2023beyond}
Aviad Rubinstein and Junyao Zhao.
\newblock Beyond worst-case budget-feasible mechanism design.
\newblock In {\em 14th Innovations in Theoretical Computer Science Conference
  (ITCS 2023)}, volume 251 of {\em Leibniz International Proceedings in
  Informatics (LIPIcs)}, pages 93:1--93:22. Schloss Dagstuhl--Leibniz-Zentrum
  f{\"u}r Informatik, 2023.

\bibitem[ST04]{spielman2004smoothed}
Daniel~A. Spielman and Shang-Hua Teng.
\newblock Smoothed analysis of algorithms: Why the simplex algorithm usually
  takes polynomial time.
\newblock {\em Journal of the ACM}, 51(3):385--463, 2004.

\bibitem[Wei16]{wei2016constant}
Dennis Wei.
\newblock A constant-factor bi-criteria approximation guarantee for
  $k$-means++.
\newblock In {\em Advances in Neural Information Processing Systems 29}, pages
  604--612, 2016.

\end{thebibliography}
\endgroup

\end{document}